\documentclass[journal]{IEEEtran}
\usepackage{stackengine}
\usepackage{blindtext}
\usepackage{graphicx}
\usepackage[utf8]{inputenc}
\usepackage{graphics} % for pdf, bitmapped graphics files
\usepackage{epsfig} % for postscript graphics files
\usepackage{mathptmx} % assumes new font selection scheme installed
\usepackage{amsmath} % assumes amsmath package installed
\usepackage{amssymb}  % assumes amsmath package installed
\usepackage{cite}
\usepackage{multicol}
\usepackage{mathtools, cuted}
\usepackage[cmintegrals]{newtxmath}
\usepackage{epstopdf}
\usepackage{graphics} % for pdf, bitmapped graphics files
\usepackage{epsfig} 

\usepackage{mathptmx} % assumes new font selection scheme installed
\usepackage{amsmath} % assumes amsmath package installed
\usepackage{amssymb}  % assumes amsmath package installed
\usepackage{cite}
\usepackage{multicol}
\usepackage{mathtools, cuted}
\usepackage[cmintegrals]{newtxmath}
\usepackage{mathtools}

\usepackage{mathptmx}
\usepackage{helvet}
\usepackage{courier}
\usepackage{subfigure}
\usepackage{type1cm}
\usepackage{titlesec}
\usepackage{epsfig} % for postscript graphics files
\usepackage{mathptmx} % assumes new font selection scheme installed
\usepackage{amsmath} % assumes amsmath package installed
\usepackage{amssymb}  % assumes amsmath package installed
\usepackage{cite}
\usepackage{multicol}
\usepackage{mathtools, cuted}
\usepackage[cmintegrals]{newtxmath}
\usepackage{makeidx}         % allows index generation
\usepackage{algorithm}% http://ctan.org/pkg/algorithm
\usepackage{algpseudocode}% http://ctan.org/pkg/algorithmicx
\DeclareMathOperator*{\argmin}{\arg\!\min}

\usepackage{multicol}        % used for the two-column index
\usepackage[bottom]{footmisc}% places footnotes at page bottom
\usepackage{caption}
\usepackage{nomencl}
\usepackage[usenames, dvipsnames]{color}

\usepackage{multicol}        % used for the two-column index
\usepackage[bottom]{footmisc}% places footnotes at page bottom
\usepackage{caption}
\usepackage{titlesec}
\usepackage{nomencl}
\usepackage[usenames, dvipsnames]{color}
\usepackage{multirow}

\usepackage{amsthm}
 
\theoremstyle{definition}

\theoremstyle{theorem}
\newtheorem{theorem}{Theorem}

\theoremstyle{remark}

\theoremstyle{proposition}

\theoremstyle{corollary}

\theoremstyle{proof}

\newtheorem{assumption}{Assumption}
\theoremstyle{assumption}

\theoremstyle{property}

\theoremstyle{lemma}

\ifCLASSINFOpdf
\else
\fi
\begin{document}
\title{Steering Opinion through Dynamic Stackelberg Optimization}

 \author{Hossein Rastgoftar% <-this % stops a space
\thanks{{\color{black}H. Rastgoftar is with the Department
of Aerospace and Mechanical Engineering, University of Arizona, Tucson,
AZ, 85721 USA e-mail: hrastgoftar@arizona.edu.}}
% \thanks{{\color{black}Authors are with ....}
 }
\maketitle

%%%%%%%%%%%%%%%%%%%%%%%%%%%%%%%%%%%%%%%%%%%%%%%%%%%%%%%%%%%%%%%%%%%%%%%%%%%%%%%%
\begin{abstract}
This paper employs the Friedkin-Johnsen (FJ) model to describe the dynamics of opinion evolution within a social network. Under the FJ framework, the society is divided into two subgroups that include stubborn agents and regular agents. The opinions of stubborn agents are not influenced by regular agents, whereas the opinions of regular agents evolve based on the opinions of their neighboring agents. 
By defining the origin as the desired collective opinion of the society, the objective of the paper is to minimize deviations from this desired opinion. To achieve this, a Stackelberg game is established between the stubborn and regular subgroups, where the opinion adjustments of the stubborn agents and the openness variables of regular agents serve as the decision variables. The proposed solution approach integrates quadratic programming and dynamic programming to optimize these decision variables at each discrete time step using forward and backward propagation.

\end{abstract}

%%%%%%%%%%%%%%%%%%%%%%%%%%%%%%%%%%%%%%%%%%%%%%%%%%%%%%%%%%%%%%%%%%%%%%%%%%%%%%%%
\section{Introduction}
The evolution of opinions in social systems has garnered significant attention from the control community in recent years. To better understand how ideas, views, and attitudes spread within groups, researchers have developed various mathematical models to describe opinion dynamics. Among these, the Friedkin-Johnsen (FJ) model \cite{zhou2024friedkin, frasca2024opinion} and the DeGroot model \cite{wu2022mixed, zhou2020two, liu2022probabilistic} are widely used to analyze the evolution of beliefs in social networks. These well-established models have been applied to a range of applications, including forecasting the outcomes of public debates, mitigating the spread of misinformation, and designing recommendation systems. Building on this foundation, this paper adapts the Friedkin-Johnsen model to characterize opinion evolution dynamics and proposes a game-theoretic framework to strategically steer community opinions toward a desired outcome.

\subsection{Related Work}
The Fredrick Johnson (FJ) model is a common method for modeling opinion evolution in social systems \cite{zhou2024friedkin, raineri2025fj, disaro2023extension,  frasca2024opinion}. The evolution of multidimensional opinions—capturing multiple interrelated topics—under the Friedkin-Johnsen (FJ) model has been investigated in \cite{parsegov2016novel, zhou2022multidimensional, xu2023linear, parsegov2016novel, zhou2022multidimensional}, providing insights into how multiple topics co-evolve within social networks. The FJ model has been used to model opinion evolution in signed networks by incorporating antagonistic and cooperative interactions \cite{yin2019signed, kalimzhanov2022co, he2020opinion}. The stability and convergence of the dynamics of the evolution of opinion in random networks are studied in \cite{wang2024final, xing2024transient, xing2024concentration, rastgoftar2025containment}. Co-evolution of opinion and action has been studied in \cite{9303954, mo2022coevolution, wang2024co, 10168221}. In \cite{kalimzhanov2022co}, the  interplay between social link polarity and the propagation of viral phenomena in signed social networks is investigated, where  a model is developed to analyze the antagonistic relationships' influence on information spread. In \cite{he2020opinion}, peer pressure and prejudice are incorporated into the evolution of opinion within signed networks.

Application of game theory in social systems have been investigated by researchers in \cite{he2020opinion, 6760259, li2011optimal, 9303954, mo2022coevolution, wang2024co, 10168221, li2023opinion, kareeva2024stackelberg, zhou2022multidimensional, zhang2011optimal, zhou2024friedkin, koshal2012stackelberg, proskurnikov2017tutorial}
 Opinion evolution was modeled as a mean-field game in \cite{6760259}. The application of control and game-theoretic methods to opinion evolution is surveyed in \cite{proskurnikov2017tutorial}. From a control and optimization perspective, \cite{li2011optimal} formulates utility-based demand response as a Stackelberg game, showing how leader-follower interactions can be leveraged for distributed optimization with bounded rationality and structural coupling—ideas that resonate with the design of our opinion evolution model. The co-evolutionary dynamics of opinion and action have been systematically investigated in \cite{9303954, mo2022coevolution, wang2024co, 10168221}. Researchers have also applied the  Stackelberg game theory to model and analyze opinion dynamics within the FJ framework to strategically perceive influences across the social network and provide some network control recommendations \cite{kareeva2024stackelberg, li2023opinion, kareeva2024stackelberg, zhou2022multidimensional}. By adapting the FJ model, \cite{zhou2024friedkin} applies the Stackelberg game to model influences between stubborn and regular agents.  In \cite{marden2007connections}, cooperative control problems are related to potential games, illustrating how game-theoretic incentives can drive distributed consensus. Similarly,  \cite{zhang2011optimal} designs optimal feedback laws for synchronization in cooperative systems, which provide design principles for aligning agent states under communication constraints.

\subsection{Contribution}

We study opinion dynamics in a society of regular and stubborn agents under the Friedkin–Johnsen (FJ) model \cite{zhou2024friedkin}. The objective is to drive the society’s average opinion toward a consensus state, defined as the origin in the opinion space, which lies within the convex hull of the stubborn agents’ opinions, while minimizing the required effort. We investigate this problem by modeling the interaction between stubborn and regular agents as a finite-horizon Stackelberg game. In this game-theoretic setting, the \emph{stubborn agents} takes actions first  to (i) minimize deviations from their initial opinions and (ii) steer the opinions of regular agents toward a desired value, assumed without loss of generality to be the origin. In response, the \emph{regular agents} seek to (i) minimize their susceptibility to external influence and (ii) align the population’s average opinion with the target. 

We develop an algorithic framework that computes the optimal strategies for both stubborn and regular sub-groups at every discrete time $k$ by iteratively combining  dynamic programming with quadratic programming. Compared to existing work, the main contributions of this paper are as follows:
\begin{enumerate}
    \item \textbf{Openness as a Decision Variable:} The paper obtains a revised form of the FJ dynamics where ``bias'' on the initial opinions is substituted by ``openness'' in accepting the new opinions. By incorporating this revision, we obtain nonlinear network opinion dynamics with the state vector aggregating opinions of regular agents and decision variables aggregating stubborn agents’ change of opinion and regular agents’ openness. Regular agents explicitly control their openness to external influence through bounded openness parameters, modeling both cognitive resistance and strategic behavior.
    \item      \textbf{Asymmetric Agent Roles:} 
    The paper captures structural heterogeneity by modeling stubborn agents as leaders who anticipate the best-response of regular agents, reflecting realistic social hierarchies and influence asymmetries.
    \item \textbf{Bi-Level Optimization Approach:} A hybrid control strategy is proposed, wherein stubborn agents solve a dynamic programming problem, and regular agents solve a constrained quadratic program at each time step. This enables decentralized yet coordinated policy synthesis.  
\end{enumerate}

\subsection{Outline}
This paper is organized as follows: The preliminary notions of opinion evolution under the FJ model is reviewed in Section \ref{Problem Statement}. Steering opinion through dynamic stackelberg optimization is formulated  in Section \ref{Problem Formulation}. Opinion evolution network dynamics is obtained in Section \ref{Opinion Evolution Dynamics} and used in Section \ref{Decentralized Acquisition of Biases and Influences} to determine optimal strategies for the regular agents to optimize theire openness and stubborn agents to minimize their change of opinions. Simulation results are presented in Section \ref{Simulation Results} and followed by the concluding remarks in Section \ref{Conclusion}.

 \section{Preliminaries} \label{Problem Statement}
We consider a network of $N$ agents indexed by $\mathcal{V} = \{1, \dots, N\}$, each holding an opinion $o_i(k) \in \mathcal{O} = [a, b]$ at discrete time $k \in \mathbb{Z}$, where $a < 0 < b$. The interaction among agents is defined by a directed graph $\mathcal{G}(\mathcal{V}, \mathcal{E})$, where $(j,i) \in \mathcal{E}$ indicates that agent $j$ influences agent $i$. Let
\begin{equation}
    \mathcal{N}_i = \{j \in \mathcal{V} : (j,i) \in \mathcal{E}\}
\end{equation}
denote the in-neighbors of agent $i$.

We define the set of \textit{stubborn agents} as those not influenced by others:
\begin{equation}
    \mathcal{V}_S = \{i \in \mathcal{V} : \mathcal{N}_i = \emptyset\},
\end{equation}
and the \textit{regular agents} as $\mathcal{V}_R = \mathcal{V} \setminus \mathcal{V}_S$, each with a time-varying \textit{openness} $\lambda_i(k) \in (0,1]$ quantifying their susceptibility to influence.

Following the Friedkin–Johnsen model \cite{zhou2024friedkin}, the opinion dynamics are given by
\begin{equation}\label{opinionevolutionindividual}
\resizebox{0.99\hsize}{!}{%
        $
o_i(k+1) = 
\begin{cases}
o_i(0) + u_i(k), & i \in \mathcal{V}_S, \\
\lambda_i(k) \sum_{j \in \mathcal{N}_i} w_{i,j}(k) o_j(k) + \left(1 - \lambda_i(k)\right) o_i(0), & i \in \mathcal{V}_R,
\end{cases}
$
}
\end{equation}
where $u_i(k)$ is the input for stubborn agents, $w_{i,j}(k) \geq 0$ is the influence weight satisfying
\begin{equation}
    \sum_{j \in \mathcal{N}_i} w_{i,j}(k) = 1.
\end{equation}

\begin{assumption}\label{assumstubborn}
    The paper assumes that there exist at least two stubborn agents where the convex hull defined by the opinions of the stubborn is an interval enclosing $0$ at every discrete time $k$.
\end{assumption}
Assumption \ref{assumstubborn} has the following implications:
\begin{enumerate}
    % \item There exist at least two stubborn agents.
    \item There exists at least one stubborn agent $i\in \mathcal{V}$ with opinion $o_i(k)<0$ at every discrete time $k$.
    \item There exists at least one stubborn agent  $j\in \mathcal{V}$  with opinion $o_j(k)>0$ at every discrete time $k$.
\end{enumerate}

% \resi

\section{Problem Formulation}\label{Problem Formulation}
Opinion evolution is formulated as a Stackelberg game between the stubborn and regular subgroups, played over a finite time horizon. Both subgroups aim to reach consensus at the origin of the opinion space. In this game, \textit{stubborn agents} ($\mathcal{V}_S$) act as leaders, selecting their control inputs first while anticipating the best-response strategies of the \textit{regular agents} ($\mathcal{V}_R$), who serve as followers. Regular agents observe the leaders’ actions and adjust their openness decisions to minimize their own cost. Specifically, agents in $\mathcal{V}_S$ aim to minimize the change in their opinions, while agents in $\mathcal{V}_R$ aim to minimize their openness to change.

Stubborn agents select $\{u_j(k)\}_{j \in \mathcal{V}_S}$ to minimize the cost
\begin{equation}
    C_k = \sum_{h=k}^{n-1} \left( \sum_{i \in \mathcal{V}_R} q_i o_i^2(h) + \sum_{j \in \mathcal{V}_S} r_j u_j^2(h) \right) + C_n,
\end{equation}
where $q_i \geq 0$, $r_j > 0$ are weighting coefficients, and $C_n$ is a quadratic terminal cost defined over the final opinions of regular agents.

The regular agents, in response, treat the openness variables as decision variables. At each discrete time $k$, they minimize the cost
\begin{equation}
    J^{'}_k = \frac{1}{2} \sum_{i \in \mathcal{V}_R} \left[ o_i(k+1) \sum_{j \in \mathcal{V}_R} o_j(k+1) + \epsilon \lambda_i^2(k) \right],
\end{equation}
subject to the constraint
\begin{equation}
    \lambda_i(k) \geq 0, \qquad \forall i \in \mathcal{V}_R,\; \forall k.
\end{equation}

To formalize this problem, Section~\ref{Opinion Evolution Dynamics} derives the network opinion dynamics, while Section~\ref{Decentralized Acquisition of Biases and Influences} formulates a quadratic program (QP) for the regular agents and a dynamic programming (DP) strategy for the stubborn agents, consistent with the Stackelberg leader–follower framework.

\section{Opinion Evolution Dynamics}\label{Opinion Evolution Dynamics}
Let $\mathcal{V}_R=\left\{1,\cdots,N_R\right\}$ and  $\mathcal{V}_S=\left\{N_R+1,\cdots,N\right\}$ identify the regular and stubborn agents, respectively. Then,
\begin{equation}
    \mathbf{x}(k)=\begin{bmatrix}o_1(k)&\cdots&o_{N_R}(k)\end{bmatrix}^T\in \mathbb{R}^{N_R\times 1}
\end{equation}
is the state vector, aggregating opinions of \textit{regular} agents.  We define the opinion change of stubborn agents as
\begin{equation}
    \mathbf{u}(k) = \mathbf{v}(k) - \mathbf{v}_0,
\end{equation}
where
\begin{equation}
    \mathbf{v}(k) = \begin{bmatrix} o_{N_R+1}(k) & \cdots & o_N(k) \end{bmatrix}^T \in \mathbb{R}^{N - N_R}
\end{equation}
is the vector of their opinions at time $k$, and $\mathbf{v}_0 = \mathbf{v}(0)$.

The inter-agent influence matrix $\mathbf{W} = [W_{i,j}] \in \mathbb{R}^{N_R \times N}$ is given by
\begin{equation}
    W_{i,j} = \begin{cases}
        \frac{1}{|\mathcal{N}_i|}, & i \in \mathcal{V}_R,~ j \in \mathcal{N}_i \\
        0, & \text{otherwise}
    \end{cases}.
\end{equation}

We define
\begin{equation}
    \mathbf{y}\left(k\right)=\begin{bmatrix}
        \lambda_1(k)&\cdots&\lambda_{N_R}(k)
    \end{bmatrix}^T\in \mathbb{R}^{N_R\times 1}
\end{equation}
as the vector aggregating the openness of regular agents. Then,
\begin{equation}\label{Lambda}
    \mathbf{\Lambda}(\mathbf{y})=\mathrm{diag}\left(\mathbf{y}\right)\in \mathbb{R}^{N_R\times N_R}
\end{equation}
is a diagonal and positive-definite matrix that is called \textit{openness matrix}. Given $\mathbf{W}$ and $\mathbf{\Lambda}$, we define
\begin{equation}
    \mathbf{L}=
        \mathbf{\Lambda}\mathbf{W}    
    \in \mathbb{R}^{N_R\times N}
    \end{equation}
and partition it as follows:
\begin{equation}\label{AB}
    \mathbf{L}\left(\mathbf{y}\right)=\begin{bmatrix}
        \mathbf{A}\left(\mathbf{y}\right)&\mathbf{B}\left(\mathbf{y}\right)
    \end{bmatrix}
    ,
\end{equation}
where $\mathbf{A}\in \mathbb{R}^{N_R\times N_R}$ and $\mathbf{B}\in \mathbb{R}^{N_R\times \left(N-N_R\right)}$. The paper also defines
\begin{equation}\label{d}
    \mathbf{d}\left(\mathbf{y}\right)=\mathbf{B}\left(\mathbf{y}\right)\mathbf{v}_0+\left(\mathbf{I}_{N_R}-\mathbf{\Lambda}\left(\mathbf{y}\right)\right)\mathbf{x}_0,
\end{equation}
where $\mathbf{x}_0=\mathbf{x}(0)$ aggregates the initial opinions of the regular agents, $\mathbf{I}_{N_R}\in \mathbb{R}^{N_R\times N_R}$ the identity matrix is. Given above definitions, the opinions of the regular agents are updated by
\begin{equation}\label{mainopiniondynamics}
\mathbf{x}\left(k+1\right)=\mathbf{A}\left(\mathbf{y}\right)\mathbf{x}\left(k\right)+\mathbf{B}\left(\mathbf{y}\right)\mathbf{u}\left(k\right)+\mathbf{d}\left(\mathbf{y}\right).
\end{equation}

\section{Optimal Strategy for Regular and Stubborn Agents}\label{Decentralized Acquisition of Biases and Influences}
In this section, we propose a game-theoretic model to interactively determine the regular agents' openness vector $\mathbf{y}(k)$ and $\mathbf{u}\left(k\right)$ quantifying the  change of opinion of the stubborn agents. 

\subsection{Optimizing Openness of Regular Agents}\label{Optimizing Openness of Regular Agents}
We propose to obtain $\mathbf{y}(k)$ by solving a QP given by
\begin{equation}\label{QPCost}
J_k\left(\mathbf{x}(k),\mathbf{y}(k),\mathbf{u}(k)\right)={1\over 2}\mathbf{y}^T\mathbf{P}\left(\mathbf{x},\mathbf{u}\right)\mathbf{y}+\mathbf{f}\left(\mathbf{x},\mathbf{u}\right)\mathbf{y}
\end{equation}
subject to inequality constraint
\begin{equation}\label{QPConstraint}
    \mathbf{Z}\mathbf{y}(k)\leq\mathbf{b},
\end{equation}
where $\mathbf{P}\in \mathbb{R}^{N_R\times N_R}$ is positive definite, $\mathbf{f}\in \mathbb{R}^{1\times N_R}$ is a row-vector, and $\mathbf{Z}\in \mathbb{R}^{2N_R\times N_R}$ and  $\mathbf{b}\in \mathbb{R}^{2N_R\times 1}$ are constant vectors. The objective of the QP is to optimize the decision vector $\mathbf{y}(k)$ at every discrete time $k$.

\subsubsection{Matrix $\mathbf{P}$ and Vector $\mathbf{f}$}
The objective of the regular agents is to minimize
\begin{equation}
\resizebox{0.99\hsize}{!}{%
        $
    J'_k(\mathbf{x}(k),\mathbf{y}(k),\mathbf{u}(k))={1\over 2}\left(\mathbf{x}^T(k+1)\mathbf{1}_{N_R\times N_R\textbf{}}\mathbf{x}(k+1)+\epsilon \mathbf{y}^T(k)\mathbf{y}(k)\right),
    $
    }
\end{equation}
where $\epsilon>0$ is a scaling parameter. The opinion dynamics \eqref{mainopiniondynamics}
\begin{equation}
\mathbf{x}(k+1)=\mathbf{H}\left(\mathbf{x}(k),\mathbf{u}(k)\right)\mathbf{y}(k)+\mathbf{x}_0,
\end{equation}
where
\begin{equation}\label{HH}
    \mathbf{H}\left(\mathbf{x}(k),\mathbf{u}(k)\right)=\mathrm{diag}\left(\mathbf{W}\begin{bmatrix}\mathbf{x}^T(k)&\mathbf{u}^T(k)+\mathbf{v}_0^T\end{bmatrix}^T-\mathbf{x}_0\right). 
\end{equation}
The cost function $J_k^{'}$ can be rewritten as follows:
\begin{equation}
     J^{'}_k\left(\mathbf{x}(k),\mathbf{y}(k),\mathbf{u}(k)\right)=J_k\left(\mathbf{x}(k),\mathbf{y}(k),\mathbf{u}(k)\right)+\mathrm{constant},
\end{equation}
where $\mathrm{constant}={1\over 2}\mathbf{x}_0^T\mathbf{x}_0$, and 
\begin{subequations}
\begin{equation}\label{P}
\mathbf{P}\left(\mathbf{x},\mathbf{u}\right)=   \mathbf{H}^T\left(\mathbf{x},\mathbf{u}\right)\mathbf{H}\left(\mathbf{x},\mathbf{u}\right)+\epsilon \mathbf{I}_{N_R} ,
\end{equation}
\begin{equation}\label{f}
   \mathbf{f}\left(\mathbf{x},\mathbf{u}\right)= \mathbf{x}_0^T\mathbf{H}\left(\mathbf{x},\mathbf{u}\right).
\end{equation}
\end{subequations}
\subsubsection{Matrix $\mathbf{Z}$ and Vector $\mathbf{b}$}
For every regular agent $i\in \mathcal{V}_R$, $\lambda_i(k)\in [0,1]$. Therefore,
\begin{subequations}
    \begin{equation}
        \mathbf{Z}=\begin{bmatrix}
            \mathbf{I}_{N_R}&-\mathbf{I}_{N_R}
        \end{bmatrix}
        ^T,
    \end{equation}
    \begin{equation}
        \mathbf{b}=\begin{bmatrix}
            \mathbf{1}_{1\times N_R}&\mathbf{0}_{1\times N_R}
        \end{bmatrix}
        ^T.
    \end{equation}
\end{subequations}
\subsection{Stubborn Agents Opinion Change Optimization}\label{Stubborn Agents Opinion Change Optimization}
We define
\begin{equation}\label{Jsk}
\begin{split}
    C_k=&{1\over2}\sum_{h=k}^{n-1}\left(\mathbf{x}^T(h)\mathbf{Q}\mathbf{x}(h)+\mathbf{u}^T(h)\mathbf{R}\mathbf{u}(h)\right)\\
    +&{1\over 2}\begin{bmatrix}
        \mathbf{x}^T(n)&
        \mathbf{d}^T(\mathbf{y}(n))
    \end{bmatrix}
    \mathbf{S}(n)\begin{bmatrix}
        \mathbf{x}^T(n)&
        \mathbf{d}^T(\mathbf{y}(n))
    \end{bmatrix}^T
\end{split}    
\end{equation}
where $\mathbf{R}\in \mathbb{R}^{\left(N-N_R\right)\times \left(N-N_R\right)}$ is positive definite, $\mathbf{Q}=[Q_{ij}]\in \mathbb{R}^{N_R\times N_R}$ is positive semi definite, and
\begin{equation}
    \mathbf{S}(n)=\begin{bmatrix}
        \mathbf{S}_{11}(n)&\mathbf{S}_{12}(n)\\
        \mathbf{S}_{21}(n)&\mathbf{S}_{22}(n)
    \end{bmatrix}
\end{equation}
is positive semi-definite. Therefore,
\begin{equation}
    \mathbf{S}_{21}(n)=\mathbf{S}_{12}^T(n).
\end{equation}
The  value function at discrete time $t=k, k+1,\cdots,n-1$ is then defined by:
\begin{equation}
    V_{k}(\mathbf{x}(t),\mathbf{y}(k))=\min_{\mathbf{u}(t),\mathbf{u}({t+1}),\cdots,\mathbf{u}({n-1})}C_k
\end{equation}
The optimal control $\mathbf{u}^*(t)$ is then obtained recursively at time $t=k,k+1,\cdots,n-1$ by solving the Bellman equation that can be expressed by
\begin{equation}
\begin{split}
     V_{k}(\mathbf{x}(t),\mathbf{y}(k))=&\min\limits_{\mathbf{u}(t)}\bigg[{1\over 2}\left(\mathbf{x}^T(t)\mathbf{Q}\mathbf{x}(t)+\mathbf{u}^T(t)\mathbf{R}\mathbf{u}(t)\right)\\
     +& V_{k}(\mathbf{x}(t+1),\mathbf{y}(k))\bigg]
\end{split}   
\end{equation}
subject to dynamics \eqref{mainopiniondynamics}.

% \begin{equation}
%     J_T={1\over2}\mathbf{x}_T^T\mathbf{S}_{1,T}\mathbf{x}_T+{1\over2}\mathbf{d}^T\mathbf{S}_{2,T}\mathbf{d}+\mathbf{s}_{3,T}\mathbf{x}_T+\mathbf{s}_{4,T}\mathbf{d}+s_{5,T},
% \end{equation}
% where $\mathbf{S}_{1,T}\in \mathbb{R}^{\left(N-N_R\right)\times \left(N-N_R\right)}$ and $\mathbf{S}_{2,T}\in \mathbb{R}^{\left(N-N_R\right)\times \left(N-N_R\right)}$ are positive semi definite, $\mathbf{s}_{3,T}\in \mathbb{R}^{1\times  \left(N-N_R\right)}$ and $\mathbf{s}_{4,T}\in \mathbb{R}^{1\times  \left(N-N_R\right)}$ are row vectors, and $s_{5,T}\in \mathbb{R}$.

\begin{algorithm}
  \caption{Solving the openness of regular agents and change of opinions of the stubborn agents.}\label{alg2}
  \begin{algorithmic}[1]
        \State \textit{Get:} Final time $n$, $\mathbf{A}\left(\mathbf{y}\right)$, $\mathbf{B}\left(\mathbf{y}\right)$, $\mathbf{d}\left(\mathbf{y}\right)$, $\mathbf{S}_{11}(n)$, $\mathbf{S}_{12}(n)$, $\mathbf{S}_{22}(n)$, $\mathbf{Q}$, $\mathbf{R}$, and $\gamma$;
        % \State \textit{Obtain:} Desired position of every agent $i\in \mathcal{V}\setminus \mathcal{W}_0$ and desired communication weights.
        \State $k=0$; 
        \State $\mathbf{v}(k)=\mathbf{v}_0$
        \State $\mathbf{x}(k)=\mathbf{x}_0$
        % \For{\texttt{ $k=0,\cdots,n$}}
        \For{\texttt{ $k=0,1,\cdots,n-1$}}
            \State $check=1$
           %  \If{$k=0$,}           
           %     \State $\mathbf{v}(k)=\mathbf{v}_0$
           % % \Else
           % \EndIf
            \State $\bar{\mathbf{y}}=\mathbf{0}_{N_R\times 1}$;
            \While{$check = 1$}
                \State Obtain $\mathbf{H}$ by \eqref{HH} given $\mathbf{x}$ and $\mathbf{v}$;
                \State Obtain $\mathbf{f}$ and $\mathbf{P}$ by \eqref{P} and \eqref{f} given $\mathbf{H}$;
                \State Solve the QP to obtain $\mathbf{y}^*(k)$ by Eqs. \label{QPCost} and \eqref{QPConstraint};    
                \If{$\left|\bar{\mathbf{y}}-\mathbf{y}^*\right|\leq \gamma\mathbf{1}_{N_R\times 1}$}
                    \State $check=0$. 
                \Else 
                    \State Obtain $\mathbf{\Lambda}\left(\mathbf{y}^*(k)\right)$ by Eq. \eqref{Lambda};
                    \State Obtain $\mathbf{A}$, $\mathbf{B}$, and $\mathbf{d}$ by \eqref{AB} and \eqref{d} given $\mathbf{y}^*$;
                    \State Solve the DP in § \ref{Stubborn Agents Opinion Change Optimization} to determine $\mathbf{u}^*(k)$;
                    \State Update $\mathbf{v}(k)$: $\mathbf{v}(k)\leftarrow \mathbf{u}^*(k)+\mathbf{u}_0$;
                    % \State Update $\mathbf{x}(k)$ by \eqref{mainopiniondynamics}: $\mathbf{x}(k)\leftarrow \mathbf{u}^*(k)+\mathbf{u}_0$;
                    \State $\bar{\mathbf{y}}\leftarrow \mathbf{y}^*(k)$;
                \EndIf
                % \State Assign centroid of $\mathcal{D}_i(t)$ ($\bar{\mathbf{d}}_i(t)$) by Eq. \eqref{followerdesired position}.    
                % \State Solve $w_{i,j}(t)$, $w_{i,k}(t)$, $w_{i,l}(t)$  by Eq. minimizing cost  \eqref{QP} subject to constraints \eqref{weq} and \eqref{wineq}.
          % \If{${1\over \mathbf{d}_j\cdot \mathbf{e}_i}\geq \alpha_i$}           
          %      \State $\alpha_i\leftarrow {1\over \mathbf{d}_j\cdot \mathbf{e}_i}$.
          %  \EndIf
           % \State Compute current $\phi_{il}\left(t\right)$ and $\psi_{il}(t)$ using Eq. \eqref{MainTransformation}.
           % \State Compute next potential  $\phi'_{il}$: $\phi_{il}=\phi_{il}+v_l\Delta t$.
           % \State Compute next stream $\psi'_{il}$: $\psi_{il}=\psi_{il}$.
           % \State Compute next $x'_i$: $ x'_i=g_1\left(\phi'_{il},\psi'_{il},\theta_l,t\right)$.
           % \State Compute next $y'_i$: $ y'_i=g_2\left(\phi'_{il},\psi'_{il},\theta_l,t\right)$.
           \EndWhile   
       \EndFor           
  \end{algorithmic}
\end{algorithm}

\begin{theorem}
    Assume that the evolution of opinions is modeled by dynamics \eqref{mainopiniondynamics} where the objective of stubborn agents is to assign their change of opinion $\mathbf{u}(k)$ such that the cost function \eqref{Jsk} is minimized as every time $k$. Given the quadratic final cost \eqref{Jsn} with a positive semi-definite matrix $\mathbf{S}(n)$, the cost-to-go is obtained by
    \begin{equation}
    \begin{split}
        V_{k}\left(\mathbf{x}(t),\mathbf{y}(t)\right)=&{1\over 2}\mathbf{x}^T(t)\mathbf{S}_{11}(t)\mathbf{x}(t)+{1\over 2}\mathbf{d}^T(\mathbf{y}(k))\mathbf{S}_{22}(t)\mathbf{d}(\mathbf{y}(k)))\\
        +&\mathbf{x}^T(t)\mathbf{S}_{12}(t)\mathbf{d}(\mathbf{y}(k))
    \end{split}        
    \end{equation}
    for $t=k, k+1,\cdots,n-1$, and 
    the optimal change of opinion $\mathbf{u}^*(t)$ is obtained by 
    \begin{equation}\label{optimalcontrol}
    \mathbf{u}^*(t)=-\mathbf{K}(t)\mathbf{x}(t)-\mathbf{H}(t)\mathbf{d},\qquad t=k,k+1,\cdots,n-1,
\end{equation}
where $\mathbf{K}(t)$, $\mathbf{H}(t)$, $\mathbf{S}_{11}(t)$,  $\mathbf{S}_{22}(t)$, and $\mathbf{S}_{12}(t)$ are updated in a backward fashion as follows:      
\begin{subequations}
     \begin{equation}\label{K}
        \mathbf{K}(t)=\left(\mathbf{B}^T\mathbf{S}_{11}(t+1)\mathbf{B}+\mathbf{R}\right)^{-1}\mathbf{B}^T\mathbf{S}_{11}(t+1)\mathbf{A},
    \end{equation}
    \begin{equation}
    \resizebox{0.99\hsize}{!}{%
        $
        \mathbf{H}(t)=\left(\mathbf{B}^T\mathbf{S}_{11}(t+1)\mathbf{B}+\mathbf{R}\right)^{-1}\mathbf{B}^T\left(\mathbf{S}_{11}(t+1)+\mathbf{S}_{12}(t+1)\right),
        $
        }
    \end{equation}
    \begin{equation}\label{H}
    \begin{split}
        \mathbf{S}_{11}(t)=&\left(\mathbf{A}-\mathbf{B}\mathbf{K}(t)\right)^T\mathbf{S}_{11}(t+1)\left(\mathbf{A}-\mathbf{B}\mathbf{K}(t)\right)\\
        +&\mathbf{K}^T(t)\mathbf{R}\mathbf{K}(t)+\mathbf{Q}
    \end{split}        
    \end{equation}
     \begin{equation}
    \begin{split}
        \mathbf{S}_{22}(t)=&\left(\mathbf{I}_{N_R}-\mathbf{B}\mathbf{H}(t)\right)^T\mathbf{S}_{11}(t+1)\left(\mathbf{I}_{N_R}-\mathbf{B}\mathbf{H}(t)\right)\\
        +&\mathbf{H}^T(t)\mathbf{R}\mathbf{H}(t)+\mathbf{S}_{22}(t+1)\\
        +&2\left(\mathbf{I}_{N_R}-\mathbf{B}\mathbf{H}(t)\right)^T\mathbf{S}_{12}(t+1)\\
        % +&\mathbf{S}_{21}(t+1)\left(\mathbf{I}_{N_R}-\mathbf{B}\mathbf{H}(t)\right)
    \end{split}        
    \end{equation}
    % \begin{equation}
    % \begin{split}
    %     \mathbf{S}_{22}(n-1)=&\left(\mathbf{I}_{N_R}-\mathbf{B}\mathbf{H}(n-1)\right)^T\mathbf{S}_{11}(n)\left(\mathbf{I}_{N_R}-\mathbf{B}\mathbf{H}(n-1)\right)\\
    %     +&\mathbf{H}^T(n-1)\mathbf{R}\mathbf{H}(n-1)+2\mathbf{S}_{12}(n)\left(\mathbf{I}_{N_R}-\mathbf{B}\mathbf{H}(n-1)\right)\\
    %     &+\mathbf{S}_{22}(n)
    % \end{split}        
    % \end{equation}
    \begin{equation}
     \resizebox{0.99\hsize}{!}{%
        $
    \begin{split}
        \mathbf{S}_{12}(n-1)=&\left(\mathbf{A}-\mathbf{B}\mathbf{K}(n-1)\right)^T\mathbf{S}_{11}(n)\left(\mathbf{I}_{N_R}-\mathbf{B}\mathbf{H}(n-1)\right)\\
        +&\left(\mathbf{I}_{N_R}-\mathbf{B}\mathbf{K}(n-1)\right)^T\mathbf{S}_{12}(n)+\mathbf{K}^T(n-1)\mathbf{R}\mathbf{H}(n-1)
    \end{split} 
    $
    }
    \end{equation}
\end{subequations}
In addition, the optimal cost is obtained by
\begin{equation}\label{Jsn}
    V_{k}\left(\mathbf{x}(t),\mathbf{y}(t)\right)={1\over 2}\begin{bmatrix}
        \mathbf{x}(t)\\
        \mathbf{d}(\mathbf{y}(k))\\
    \end{bmatrix}^T
    \begin{bmatrix}
        \mathbf{S}_{11}(t)&\mathbf{S}_{12}(t)\\
        \mathbf{S}_{21}(t)&\mathbf{S}_{22}(t)
    \end{bmatrix}\begin{bmatrix}
        \mathbf{x}(t)\\
        \mathbf{d}(\mathbf{y}(k))\\
    \end{bmatrix}
\end{equation}
for every $t=k,k+1,\cdots,n$, where $\mathbf{S}_{21}=\mathbf{S}_{12}^T$.
% \begin{subequations}
   
% \end{subequations}
\end{theorem}
\begin{proof}
By substituting $k=n-1$, we obtain $J(n-1)$ as follows:
\begin{equation}\label{2ndlastcost}
\begin{split}
     C_k(n-1)=&{1\over2}\left(\mathbf{x}^T(n-1)\mathbf{Q}\mathbf{x}(n-1)+\mathbf{u}^T(n-1)\mathbf{R}\mathbf{u}(n-1)\right)\\
     +&\left(\mathbf{A}\mathbf{x}(n-1)+\mathbf{B}\mathbf{u}(n-1)+\mathbf{d}\right)^T\times \\
     &\mathbf{S}_{11}(n)\left(\mathbf{A}\mathbf{x}(n-1)+\mathbf{B}\mathbf{u}(n-1)+\mathbf{d}\right)\\
     +&{1\over2}\left[\mathbf{d}^T\mathbf{S}_{22}(n)\mathbf{d}+\mathbf{x}^T(n)\mathbf{S}_{12}(n)\mathbf{d}+\mathbf{d}^T\mathbf{S}_{21}(n)\mathbf{x}(n)\right]
\end{split}   
\end{equation}
% \begin{equation}\label{2ndlastcost}
% \begin{split}
%      C_k(n-1)=&{1\over2}\left(\mathbf{x}^T(n-1)\mathbf{Q}\mathbf{x}(n-1)+\mathbf{u}^T(n-1)\mathbf{R}\mathbf{u}(n-1)\right)\\
%      +&\left(\mathbf{A}\mathbf{x}(n-1)+\mathbf{B}\mathbf{u}(n-1)+\mathbf{d}\right)^T\times \\
%      &\mathbf{S}_{11}(n)\left(\mathbf{A}\mathbf{x}(n-1)+\mathbf{B}\mathbf{u}(n-1)+\mathbf{d}\right)\\
%      +&{1\over2}\mathbf{d}^T\mathbf{S}_{22}(n)\mathbf{d}+\mathbf{d}^T\mathbf{S}_{12}(n)\mathbf{x}(n)
% \end{split}   
% \end{equation}
% \begin{equation}\label{2ndlastcost}
% \begin{split}
%      J(T-1)=&{1\over2}\left(\mathbf{x}^T(T-1)\mathbf{Q}\mathbf{x}(T-1)+\mathbf{u}^T(T-1)\mathbf{R}\mathbf{u}(T-1)\right)+\mathbf{s}_{3,T}\mathbf{x}_T+\mathbf{s}_{4,T}\mathbf{d}+s_{5,T}\\
%      +&{1\over2}\mathbf{d}^T\mathbf{S}_{2,T}\mathbf{d}+\left(\mathbf{A}\mathbf{x}_{T-1}+\mathbf{B}\mathbf{u}_{T-1}+\mathbf{d}\right)^T\mathbf{S}_{1,T}\times\\
%      &\left(\mathbf{A}\mathbf{x}_{T-1}+\mathbf{B}\mathbf{u}_{T-1}+\mathbf{d}\right)
% \end{split}   
% \end{equation}
The optimal control $\mathbf{u}^*(n-1)$ is obtained by solving ${\partial C_k(n-1)\over \partial \mathbf{u}(n-1)}=\mathbf{0}$ as follows:
\begin{equation}\label{optimalcontrol}
    \mathbf{u}^*(n-1)=-\mathbf{K}(n-1)\mathbf{x}(n-1)-\mathbf{H}(n-1)\mathbf{d},
\end{equation}
where
\begin{subequations}
    \begin{equation}
        \mathbf{K}(n-1)=\left(\mathbf{B}^T\mathbf{S}_{11}(n)\mathbf{B}+\mathbf{R}\right)^{-1}\mathbf{B}^T\mathbf{S}_{11}(n)\mathbf{A},
    \end{equation}
    \begin{equation}
    \resizebox{0.99\hsize}{!}{%
        $
        \mathbf{H}(n-1)=\left(\mathbf{B}^T\mathbf{S}_{11}(n)\mathbf{B}+\mathbf{R}\right)^{-1}\mathbf{B}^T\left(\mathbf{S}_{11}(n)+{1\over2}\left(\mathbf{S}_{12}(n)+\mathbf{S}_{21}^T(n)\right)\right).
        $
        }
    \end{equation}
\end{subequations}
By substituting $\mathbf{u}_{n-1}^*$ from \eqref{optimalcontrol} into Eq. \eqref{2ndlastcost}, 
\begin{equation}
    \mathbf{x}\left(n\right)=\left(\mathbf{A}-\mathbf{B}\mathbf{K}(n-1)\right)\mathbf{x}\left(n-1\right)+\left(\mathbf{I}_{N_R}-\mathbf{B}\mathbf{H}(n-1)\right)\mathbf{d}
\end{equation}
and the {\color{black}value function} at time $n-1$ is obtained as follows:
\begin{equation}
\resizebox{0.99\hsize}{!}{%
        $
\begin{split}
   {\color{black}V_k(\mathbf{x}(n-1),\mathbf{y}(k))}=&{1\over2}\big[\mathbf{x}^T(n-1)\mathbf{S}_{11}(n-1)\mathbf{x}(n-1)+\mathbf{d}^T\mathbf{S}_{22}(n-1)\mathbf{d}\\
    +&\mathbf{x}^T(n-1)\mathbf{S}_{12}(n-1)\mathbf{d}+\mathbf{d}^T\mathbf{S}_{21}(n-1)\mathbf{x}(n-1)\big],
\end{split}  
$
}
\end{equation}
% \begin{equation}
% \begin{split}
%     J^*_S(n-1)=&{1\over2}\big[\mathbf{x}^T(n-1)\mathbf{S}_{11}(n-1)\mathbf{x}(n-1)+\mathbf{d}^T\mathbf{S}_{22}(n-1)\mathbf{d}\\
%     +&2\mathbf{d}^T\mathbf{S}_{12}(n-1)\mathbf{x}(n-1)\big],
% \end{split}  
% \end{equation}
where
\begin{subequations}
    \begin{equation}\label{s11}
    \begin{split}
        \mathbf{S}_{11}(n-1)=&\left(\mathbf{A}-\mathbf{B}\mathbf{K}(n-1)\right)^T\mathbf{S}_{11}(n)\left(\mathbf{A}-\mathbf{B}\mathbf{K}(n-1)\right)\\
        +&\mathbf{K}^T(n-1)\mathbf{R}\mathbf{K}(n-1)+\mathbf{Q}
    \end{split}        
    \end{equation}
     \begin{equation}\label{s22}
    \begin{split}
        \mathbf{S}_{22}(n-1)=&\left(\mathbf{I}_{N_R}-\mathbf{B}\mathbf{H}(n-1)\right)^T\mathbf{S}_{11}(n)\left(\mathbf{I}_{N_R}-\mathbf{B}\mathbf{H}(n-1)\right)\\
        +&\mathbf{H}^T(n-1)\mathbf{R}\mathbf{H}(n-1)+\mathbf{S}_{22}(n)\\
        +&\left(\mathbf{I}_{N_R}-\mathbf{B}\mathbf{H}(n-1)\right)^T\mathbf{S}_{12}(n)\\
        +&\mathbf{S}_{21}(n)\left(\mathbf{I}_{N_R}-\mathbf{B}\mathbf{H}(n-1)\right)
    \end{split}        
    \end{equation}
    % \begin{equation}
    % \begin{split}
    %     \mathbf{S}_{22}(n-1)=&\left(\mathbf{I}_{N_R}-\mathbf{B}\mathbf{H}(n-1)\right)^T\mathbf{S}_{11}(n)\left(\mathbf{I}_{N_R}-\mathbf{B}\mathbf{H}(n-1)\right)\\
    %     +&\mathbf{H}^T(n-1)\mathbf{R}\mathbf{H}(n-1)+2\mathbf{S}_{12}(n)\left(\mathbf{I}_{N_R}-\mathbf{B}\mathbf{H}(n-1)\right)\\
    %     &+\mathbf{S}_{22}(n)
    % \end{split}        
    % \end{equation}
    \begin{equation}\label{symmery12}
     \resizebox{0.99\hsize}{!}{%
        $
    \begin{split}
        \mathbf{S}_{12}(n-1)=&\left(\mathbf{A}-\mathbf{B}\mathbf{K}(n-1)\right)^T\mathbf{S}_{11}(n)\left(\mathbf{I}_{N_R}-\mathbf{B}\mathbf{H}(n-1)\right)\\
        +&\left(\mathbf{I}_{N_R}-\mathbf{B}\mathbf{K}(n-1)\right)^T\mathbf{S}_{12}(n)+\mathbf{K}^T(n-1)\mathbf{R}\mathbf{H}(n-1)
    \end{split} 
    $
    }
    \end{equation}
      \begin{equation}\label{symmery21}
     \resizebox{0.99\hsize}{!}{%
        $
    \begin{split}
        \mathbf{S}_{21}(n-1)=&\left(\mathbf{I}_{N_R}-\mathbf{B}\mathbf{H}(n-1)\right)^T\mathbf{S}_{11}(n)\left(\mathbf{A}-\mathbf{B}\mathbf{K}(n-1)\right)\\
        +&\mathbf{S}_{21}(n)\left(\mathbf{I}_{N_R}-\mathbf{B}\mathbf{K}(n-1)\right)+\mathbf{H}^T(n-1)\mathbf{R}\mathbf{K}(n-1)
    \end{split} 
    $
    }
    \end{equation}
\end{subequations}
It is seen that $\mathbf{S}_{12}(n-1)=\mathbf{S}_{21}^T(n-1)$ per Eqs. \eqref{symmery12} and \eqref{symmery21}. Therefore, ${\color{black}\mathbf{S}}(n-1)$ is symmetric. By substituting $n-1$ and $n$ by $t$ and $t+1$, respectively, $\mathbf{S}_{11}(t)$, $\mathbf{S}_{12}(t)$, and $\mathbf{S}_{22}(t)$ are determined at any time $t=k, k+1, \cdots, n-1$ in a backward fashion.
Given $\mathbf{S}_{11}(t)$, $\mathbf{S}_{12}(t)$, and $\mathbf{S}_{22}(t)$  we can determine $\mathbf{K}(t)$ and $\mathbf{H}(t)$, at any time $t=k, k+1, \cdots, n-1$, by using Eqs. \eqref{K} and \eqref{H}, respectively. Given $\mathbf{x}(k)$, aggregating opinions of the regular agents at time $k$, and $\mathbf{K}$ and $\mathbf{H}(k)$,  the optimal change of opinion of the stubborn agents, denoted by $\mathbf{u}(k)$, is assigned by Eq. \eqref{optimalcontrol}.
\end{proof}

\subsection{Solution}
To solve the Stackelberg game, we adopt a bi-level optimization strategy where stubborn agents (leaders) and regular agents (followers) iteratively compute their optimal strategies through backward and forward recursions, respectively.

Given the state vector $\mathbf{x}(k)$ and the optimal openness vector $\mathbf{y}^*(k)$, the stubborn agents determine their optimal opinion update $\mathbf{u}^*(k)$ by solving the Bellman equation in a backward manner:
\begin{equation}\label{g1}
\begin{split}
     \mathbf{u}^*(k) = &\argmin\limits_{\mathbf{u}(k)} \bigg\{ \frac{1}{2} \left[ \mathbf{x}^T(k) \mathbf{Q} \mathbf{x}(k)+\mathbf{u}^T(k) \mathbf{R} \mathbf{u}(k) \right] \\
     + &V_k\left(\mathbf{x}(k+1), \mathbf{y}^*(k)\right) \bigg\},
\end{split}
\end{equation}
where $V_k(\cdot)$ is the value function capturing the future cost-to-go from time $k+1$.

Conversely, given $\mathbf{x}(k)$ and the leader's optimal action $\mathbf{u}^*(k)$, the regular agents update their openness vector $\mathbf{y}^*(k)$ by solving the following quadratic program in a forward fashion:
\begin{equation}\label{g2}
\mathbf{y}^*(k) = \argmin\limits_{\mathbf{y}(k)} J_k\left(\mathbf{x}(k), \mathbf{y}(k), \mathbf{u}^*(k)\right)
\end{equation}
subject to the constraint \eqref{QPConstraint}, which ensures admissible openness values.

We employ a forward-backward iterative scheme (Algorithm~\ref{alg2}) to compute the Stackelberg equilibrium at each time step $k = 0, \dots, n-1$, alternating between follower (forward) and leader (backward) updates until convergence.

%  For given state vector $\mathbf{x}(k)$ and the optimal openness vector $\mathbf{y}^*(k)$, the stubborn agents solve the optimal change of opinion in a backward manner by solving the Bellman equation as follows:
% \begin{equation}\label{g1}
% \begin{split}
%      \mathbf{u}^*(k)=&\argmin\limits_{\mathbf{u}(k)}\bigg[{1\over 2}\left(\mathbf{x}^T(k)\mathbf{Q}\mathbf{x}(k)+\mathbf{u}^T(k)\mathbf{R}\mathbf{u}(k)\right)\\
%      +& V_{k}(\mathbf{x}(k+1),\mathbf{y}^*(k))\bigg]
% \end{split}   
% \end{equation}
% On the other hand, given the state vector $\mathbf{x}(k)$ and {\color{black}the optimal} control vector $\mathbf{u}^*(k)$,  the regular agents optimize their openness in a forward way by solving the following QP:
% \begin{equation}\label{g2}
% \mathbf{y}^*(k)=\argmin\limits_{\mathbf{y}_k}J_k\left(\mathbf{x}(k),\mathbf{y}(k),\mathbf{u}^*(k)\right)
% \end{equation}
% subject to the inequality constraint \eqref{QPConstraint}. 
% We use Algorithm \ref{alg2} to iteratively obtain $\mathbf{y}^*(k)$ and $\mathbf{u}^*(k)$ at every discrete time $k$ ($k=0,\cdots,n-1$) by forward and backward sweeps.

\section{Simulation Results}\label{Simulation Results}
% \section{Simulation Results}
For simulation, we consider the evolution of opinions in a society consisting of $N=98$ agents defined by $\mathcal{V}=\left\{1,\cdots,98\right\}$. There exist $N_R=96$ regular agents and two stubborn agents, therefore, $\mathcal{V}=\mathcal{V}_R\cup\mathcal{V}_S$, where $\mathcal{V}_R=\left\{1,\cdots,96\right\}$ and $\mathcal{V}_S=\left\{97,98\right\}$ define the regular and stubborn agents, respectively. Regular  agents are initially distributed over the interval $\left[o_{97}(0),o_{98}(0)\right]$, where $\mathbf{v}_0=\begin{bmatrix}o_{97}(0)&o_{98}(0)\end{bmatrix}^T=\begin{bmatrix}-1&1\end{bmatrix}^T$.

\begin{figure}[ht]
\centering
\includegraphics[width=0.45 \textwidth]{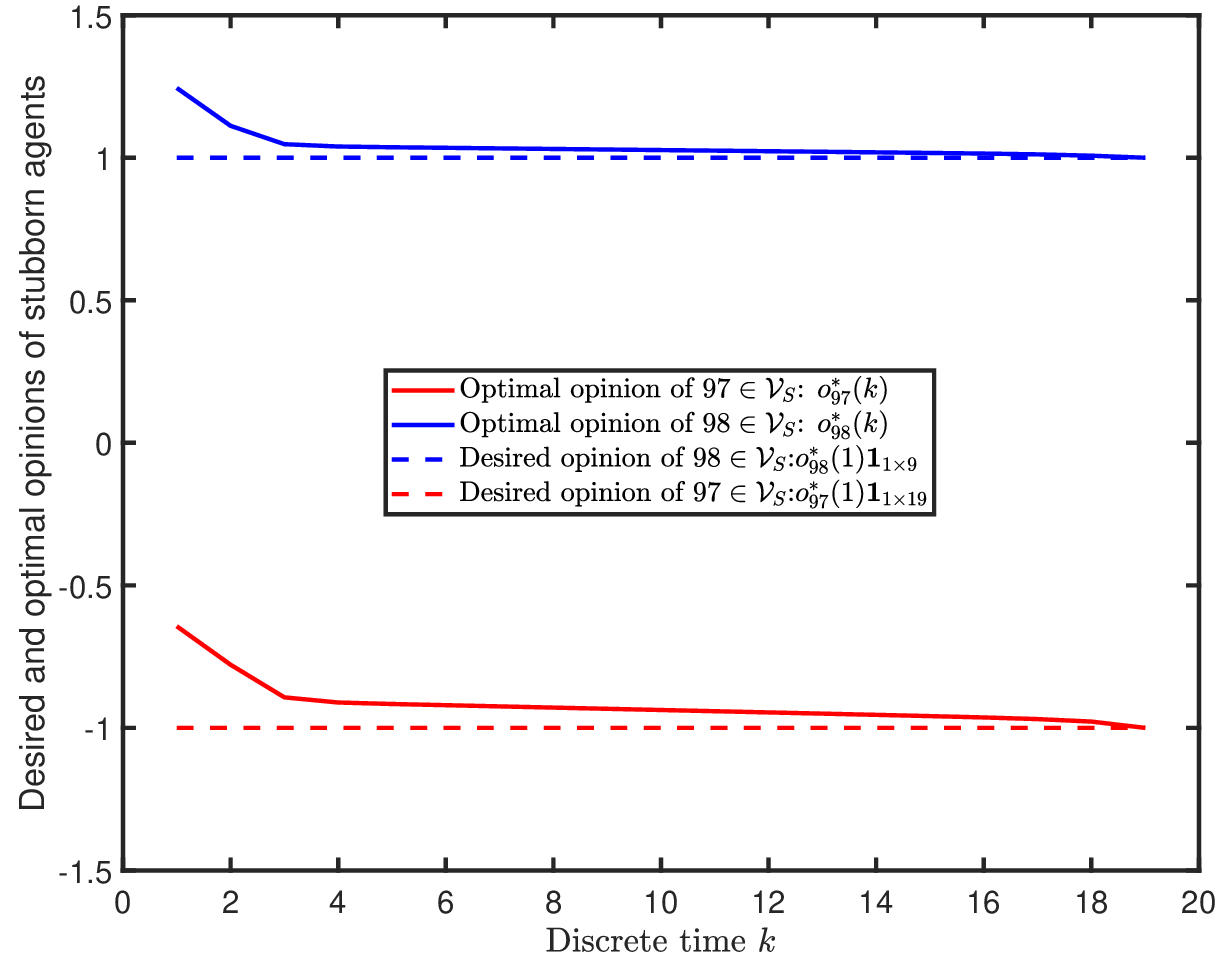}
\caption{Optimal and desired opinions of stubborn agents versus dicrete time $k$ for $k=1,\cdots,19$.}
\label{StubbornOpinions} 
\end{figure}

\begin{figure}[ht]
\centering
\includegraphics[width=0.45 \textwidth]{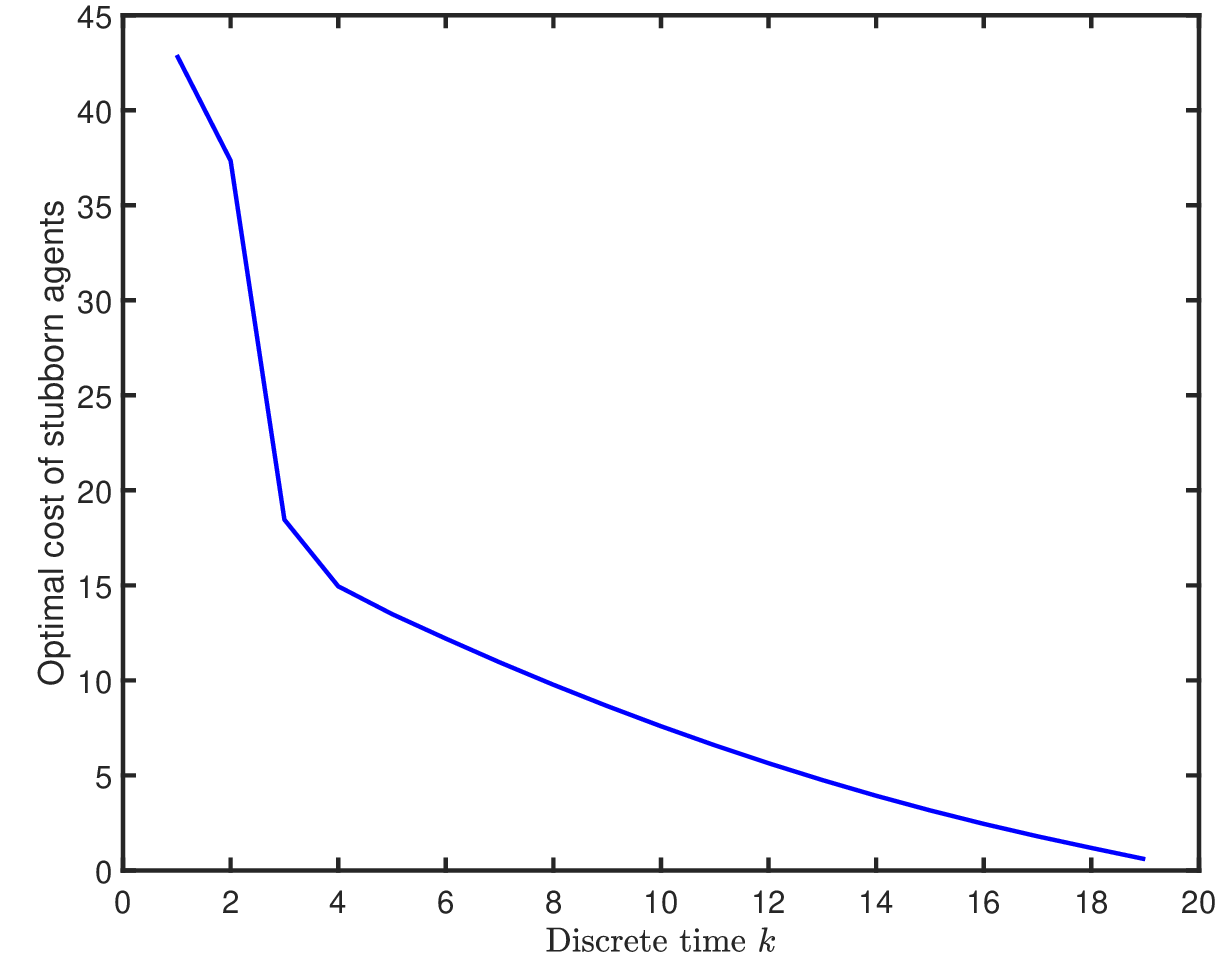}
\caption{Optimal cost $\mathbf{J}_S^*(k)$  for  $k=1,\cdots,19$.}
\label{OptimalCost} 
\end{figure}

\begin{figure}[ht]
\centering
\includegraphics[width=0.45 \textwidth]{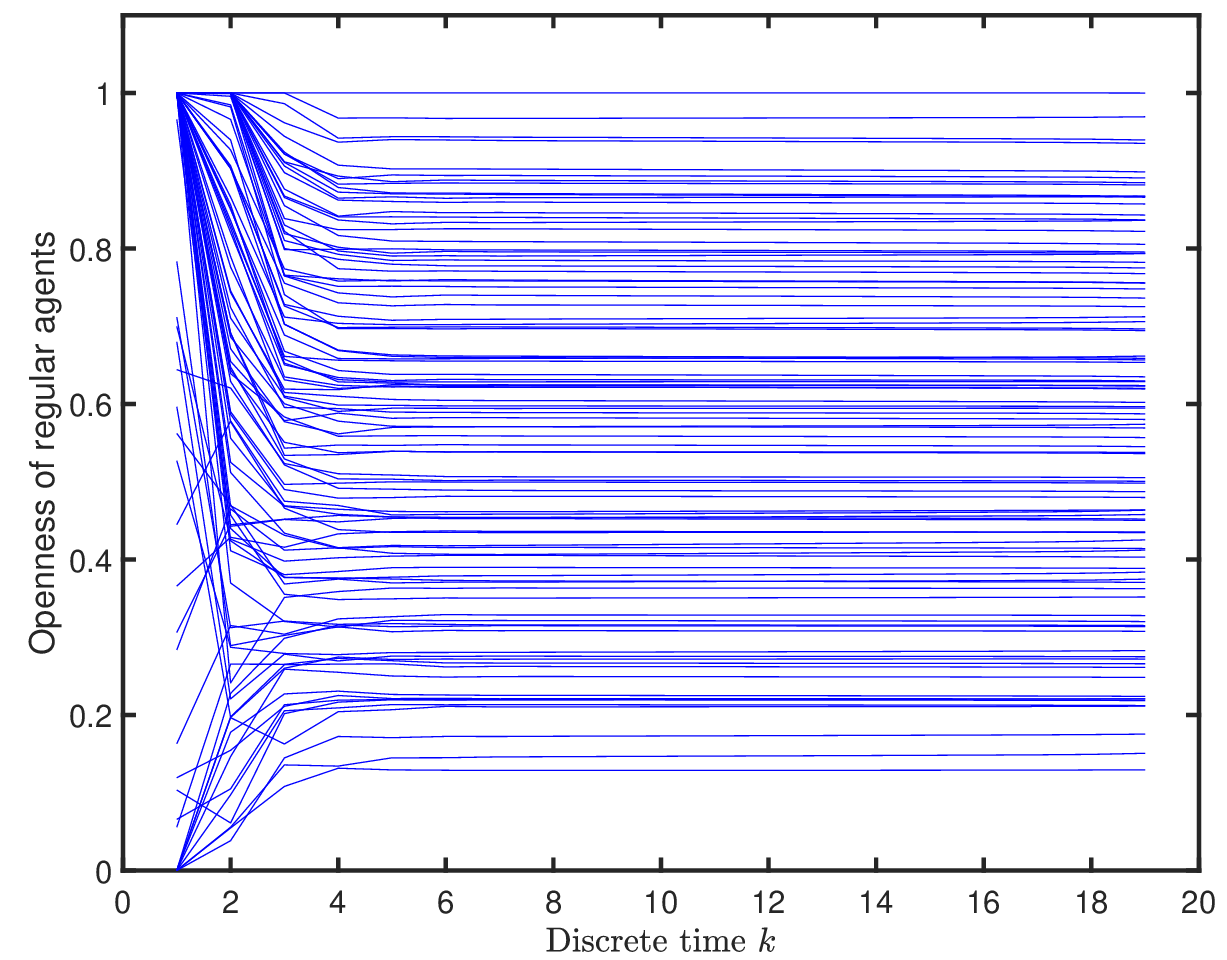}
\caption{Openness of regular agents versus discrete time $k$ for  $k=1,\cdots,19$.}
\label{Openness} 
\end{figure}

\begin{figure}[ht]
\centering
\includegraphics[width=0.45 \textwidth]{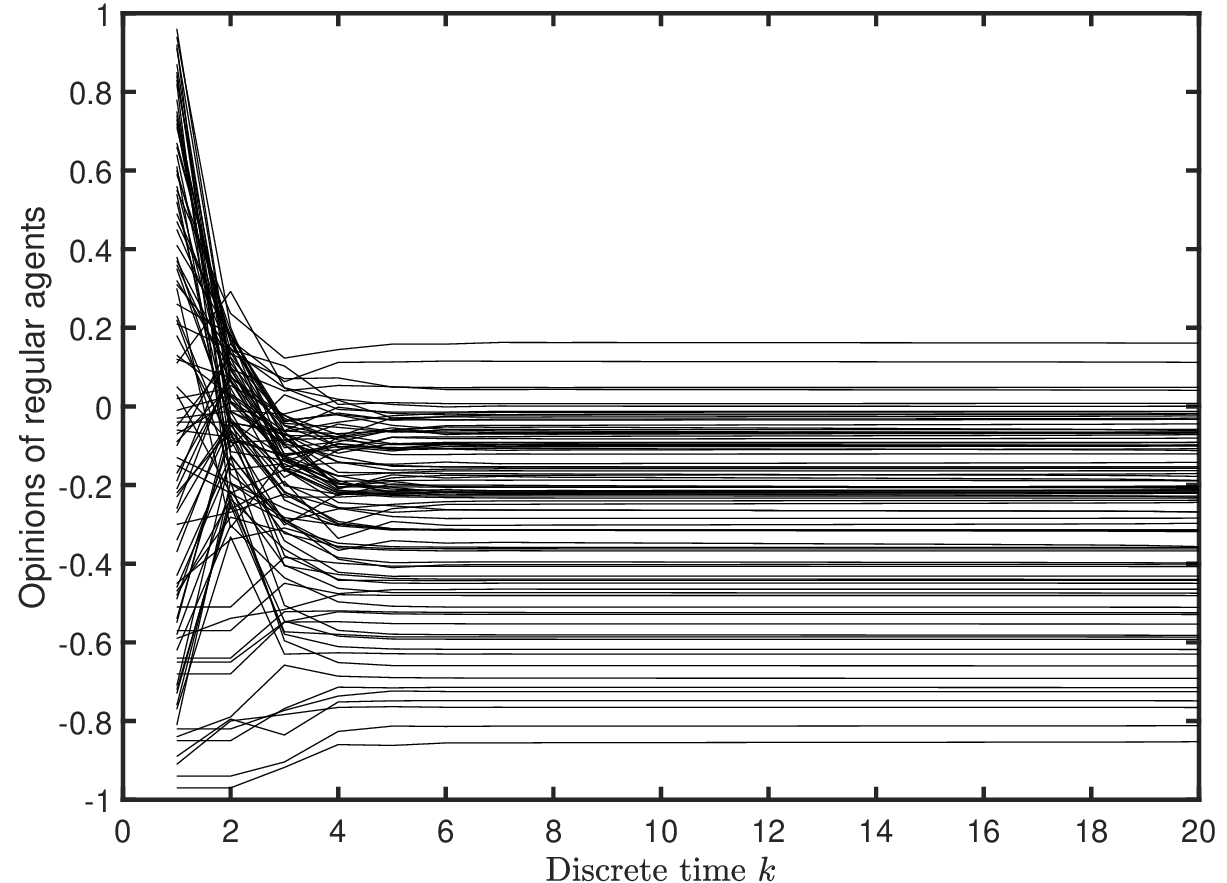}
\caption{Opinions of regular agents $1$ through $96$ versus discrete time $k$ for $k=1,\cdots,20$.}
\label{Opinions} 
\end{figure}

For simulation, we choose $\mathbf{P}=\mathbf{1}_{N_R\times N_R}$, $\epsilon=1$, $\mathbf{Q}=0.1\mathbf{I}_{N_R}$, and $\mathbf{R}=\mathbf{I}_2$. We also choose $\mathbf{S}_{11}(n)=\mathbf{S}_{12}(n)=\mathbf{S}_{22}(n)=\mathbf{0}_{96\times 96}$. Figure \ref{StubbornOpinions} shows the desired opinions of stubborn agents $97$ and $98$, represented by dashed red and blue lines, respectively. The corresponding optimal opinions of the agents $97$ and $98$ are shown by solid red and blue lines.

The optimal cost, ${J}_S^*(k)$, is plotted against discrete time steps $k = 1, \dots, 19$ in Figure \ref{OptimalCost}. The results indicate that the optimal cost converges to zero.

Figure \ref{Openness} presents the optimal openness of the regular agents as a function of discrete time $k$ over the same interval. Furthermore, the opinions of all regular agents are plotted in discrete time in Figure \ref{Opinions}.

\section{Conclusion}\label{Conclusion}
The main contribution of this paper is to develop a game-theoretical model for steering the community’s opinion towards a desired idea (opinion). By dividing a community into stubborn and regular subgroups, the paper applied the FJ model to obtain the opinion evolution dynamics and used the Stackelberg game model to tackle the opinion steering problem. This resulted in a model-based game-theoretic approach that ensures the average opinion of a community can become ultimately closer to an opinion of desire while interests of both stubborn and regular sub-groups are respected and incorporated. For this study, the paper assumes that the stubborn subgroup aims to minimize the change of their opinion as opposed to the regular sub-group that desires to minimize the group openness, while both sub-groups have consensus on imposing costs on deviation from the desired opinion which was defined as the origin of the opinion space, without loss of generality. 
% This paper models the FJ opinion evolution dynamics as a containment control problem; stubborn agents are considered as leaders specifying a desired opinion distribution of the society. It was proven and demonstrated that regular agents can achieve a desired opinion distribution in a truly decentralized fashion through random communication with some other (regular or stubborn) agents. The paper assumed that every regular agent is completely open to following the reward maximization rule suggested in Section \ref{Decentralized Acquisition of Biases and Influences} to assign its communication weight  based on the rewards acquired by its own in-neighbors (influences). For the future work, we will relax this assumption and develop a game-theoretic framework to develop a game-theoretic framework to maximize strategies for both regular and obstinate agents. This will allow the stubborn agents to modify the reward distribution in response to the regular agents' responses. 
% \section{Acknowledgement}
%  This work was supported by the
% National Science Foundation under Award 2133690 and Award 1914581.
\bibliographystyle{IEEEtran}
\bibliography{reference}
% \begin{IEEEbiography}[{\includegraphics[width=1in,height=1.25in,clip,keepaspectratio]{Rastgoftar.jpg}}]{\textbf{Hossein Rastgoftar}} is an Assistant Professor at the University of Arizona. Prior to this, he was an adjunct Assistant Professor at the University of Michigan from 2020 to 2021. He was also an Assistant Research Scientist (2017 to 2020) and a Postdoctoral Researcher (2015 to 2017) in the Aerospace Engineering Department at the University of Michigan Ann Arbor. He received the B.Sc. degree in mechanical engineering-thermo-fluids from Shiraz University, Shiraz, Iran, the M.S. degrees in mechanical systems and solid mechanics from Shiraz University and the University of Central Florida, Orlando, FL, USA, and the Ph.D. degree in mechanical engineering from Drexel University, Philadelphia, in 2015. 
% % His current research interests include dynamics and control, multiagent systems, cyber-physical systems, and optimization and Markov decision processes.
% \end{IEEEbiography}

\end{document}